\documentclass[11pt]{article}

\newtheorem{theorem}{Theorem:}
\newtheorem{proposition}{Proposition:}
\newtheorem{definition}{Definition:}

\newenvironment{proof}{\textbf{Proof:}}{\hfill $\Box$}

\usepackage{a4wide}
\usepackage{amsmath}
 
\newcommand{\RR}{\mathbb{R}}
\newcommand{\CC}{\mathbb{C}}

\newcommand{\suprime}[1]{{\LARGE To be Done}}
\newcommand{\Un}{{\mbox{U}(n)}}
\newcommand{\iii}{\boldsymbol{\imath}}

\newcommand{\trace}{\mbox{trace}}

\newcommand{\nbar}{{\overline n}}

\newcommand{\ket}[1]{\left|#1\right\rangle}
\newcommand{\bra}[1]{\left\langle #1\right|}

\newcommand{\bbibitem}[2][]{\hspace{10pt}}

\usepackage{latexsym}
 \usepackage{amsfonts,amssymb}

\usepackage{graphicx}
\usepackage{psfrag,epsfig}
\usepackage{epstopdf}

\usepackage{xcolor}
\usepackage{hyperref}

 \newcommand{\BM}[1]{\left[ \begin{array}{#1}}
 \newcommand{\EM}{\end{array} \right]}

\title 
{
Gate generation for open quantum systems via   a monotonic algorithm with  time optimization}

\author 
{ Paulo S\'ergio PEREIRA DA SILVA$^\sharp$ \\
 Pierre ROUCHON$^\S$ \\
{$^\sharp$ University of S\~ao Paulo -- USP -- Brazil}\\
{$\S$ Laboratoire de Physique de l’Ecole normale sup\'{e}rieure,}\\
{ Mines Paris-PSL, Inria, ENS-PSL, Universit\'{e} PSL, CNRS, Paris. }}

\begin{document} 
\bibliographystyle{plain}
\maketitle
\begin{abstract}
 We present a monotonic numerical algorithm including time optimization for generating quantum gates  for open systems. Such  systems  are assumed to be  governed   by Lindblad master equations for the  density operators on a large  Hilbert-space whereas the quantum gates are relative  to a sub-space of small  dimension.  Starting from an initial seed of the control input,  this algorithm consists in the repetition  of  the following two steps producing a new control input: (A) backwards integration of adjoint Lindblad-Master equations (in the Heisenberg-picture) from a set of final conditions encoding the quantum gate to generate;  (B)   forward integration of    Lindblad-Master equations in closed-loop where  a Lyapunov based control produced the new control input. The numerical stability is ensured by the stability of both the open-loop adjoint  backward system and the  forward closed-loop system. A clock-control input  can be added to the usual control input. The obtained  monotonic algorithm  allows then to optimise  not only the shape of the control imput, but also the gate time.   Preliminary numerical implementations   indicate that this algorithm  is well suited for cat-qubit gates, where  Hilbert-space dimensions (2 for the Z-gate and 4 for the CNOT-gate)  are much smaller than the dimension  of the  physical Hilbert-space involving mainly Fock-states (typically 20 or  larger for a single cat-qubit). This monotonic algorithm,  based on Lyapunov control techniques,  is shown to have  a straightforward  interpretation in terms of optimal control: its  stationary conditions  coincides with  the first-order optimality   conditions  for a  cost depending linearly  on the final values of the quantum states.

\end{abstract}
  \textbf{Acknowledgment:} This project has received funding from the European Research Council (ERC) under the European Union's Horizon 2020 research and innovation programme (grant agreement No. [884762]).

\section{Introduction}
\label{sIntro}

In theory, optimal control could provide fast solutions for state preparation and quantum gate generation \cite{PalaoK2002PRL,PalaoK2003PRA,GlaserBoscainCalarcoEtAl2015}.
The complexity of implementation of optimal control grows faster with the Hilbert dimension $n$ \cite{SchirF2011NJoP}). This motivates the research of other methods like Lyapunov stabilization, that can tackle large  dimensions   \cite{L4,L3,L1,L2,L5,L6,LS1,SilPerRou14,SilPerRou16}.
Note that Lyapunov stabilization of quantum systems appears also in the infinity-dimensional context (for instance for ensemble control of half-spin systems) \cite{BEAPERROU12,MACPERROU22}.

Several numerical algorithms have been developed for tackling quantum control. We shall call by \emph{Piecewise-Constant} case for the algorithms that provide piecewice-constant control pulses and by \emph{Smooth-Case } for the algorithms that provide smooth control inputs.
For instance, for  the Piecewise-Constant case one may consider the {Krotov} method  \cite{SchirF2011NJoP},
GRAPE (of first and second orders) \cite{KHANEJA2005,SecondGRAPE},
CRAB \cite{CRAB}, and the piecewise constant algorithm based on Lyapunov techniques \cite{PerSilRou19}. For the \emph{smooth} case, one may consider
{GOAT} \cite{MacShaTanFra15}, and the Matlab open code available for   RIGA (\cite{CODE_OCEAN_SMOOTH}. The Krotov method has also a \emph{Smooth-Case} version, called here simply by {\emph Krotov} method,  that is strongly related to algorithm that is presented in this paper. The reader may refer to the survey papers \cite{Koch2016,Goerz2014}  for the description of Krotov method. One may say that the algorithm presented in this work (without clock control) is very close to the  \emph{Krotov} method, at least in the case of the so called \emph{sequential update} of the
control, which ensures a monotonic behaviour of such method\footnote{ \emph{Sequential update} of the control means that the control pulses that are applied to the system in a step $\ell$ of the algorithm are updated ``on the fly'', that is, not only in the end of each step as is done in the classic Krotov method and also in GRAPE.}. To be monotonic in this case is a property that is analogous to the non-increasing property of the Lyapunov function in the context of the algorithm that is presented in this paper. The  contributions are
  \begin{itemize}
    \item to generalise such  monotonic algorithms by considering the optimisation of the shape of the control input and the gate time simultaneously (see  section \ref{sClock}).
    \item  to implement   such generalisation  on physical  case-studies of bosonic qubits (see section~\ref{sExample})  where the dimension of the underlying Hilbert-space is far much larger (578  in the numerical computations of  figures~\ref{FC} and~\ref{FD}) than the  size  of the orthonormal sets defining the gate ($4$ for a CNOT gate between two cat-qubits).
  \end{itemize}

Particularly relevant for the present work is the algorithm  \emph{RIGA (Reference Input Generation Algorithm)}\cite{CODE_OCEAN_SMOOTH} that  generates gate of small dimension $\bar n$  for closed quantum systems in  an Hilbert-space of larger  dimension $n\geq \bar n$   governed  by a  Schrodinger dynamics:
\begin{equation}
\label{eSh}
 \dot X(t) = -\iii (H_0 + \sum_{k=1}^{m} u_k  H_k) X(t)
\end{equation}
where the $H_k$ are the Hermitian operators, $u_k$ are scalar  control inputs    and $X(t) \in \Un$ is the propagator. The quantum gate in this case is represented by some set  of initial vectors  $\{|e_i\rangle, i=1, \ldots, \nbar \}$ and  a set of final vectors $\{|f_i\rangle, i=1, \ldots, \nbar \}$, both orthonormal subsets of $\CC^{n}$, with $\nbar \leq n$. The  gate generation relies on finding a  control input $u=(u_1,\ldots,u_m)$  steering from initial value $X(0) = I$ to  a final value $X_{goal}$ where
$X_{goal} |e_i\rangle = |f_i\rangle$, for $i=1, \ldots, \nbar$, up to some error that is measured by the so called \emph{gate fidelity}. This is equivalent
to the following steering problem:
\begin{definition}
\label{Def1}
 Let $\{|e_i\rangle, i=1, \ldots, \nbar \}$ and $\{|f_i\rangle, i=1, \ldots, \nbar \}$ be two orthonormal subsets of $\CC^{n}$.
 The problem of quantum gate generation is to find  a gate time $T_f >0$ and a time-varying   control input $u:[0, T_f] \rightarrow \RR^m$ such that
  the solution of~\eqref{eSh} starting form $X(0)=I$ verifies $X(T_f) |e_i\rangle=|f_i\rangle$  for $i=1, \ldots, \nbar$ up to some admissible error called  gate-fidelity.
\end{definition}

For a prescribed gate time $T_f$ and  gate,   RIGA is a monotonic algorithm  improving  the steering control input $[0,T_f]\ni t \mapsto u(t)$ from an initial guess $[0,T_f]\ni t \mapsto \overline{u}^0(t)$. .  Each step $\ell$ of RIGA is  as follows. Given the control input ${\overline u}^{\ell-1}(t)$ defined on $[0, T_f]$, one obtains a reference trajectory of the propagator $X(t)$  by integrating the system backwards from the final condition $X_{goal}$. A Lyapunov based  tracking control is then implemented, and a tracking control ${\overline u}^{\ell}$ is obtained by integrating forwards  the closed loop system from $X(0)=I$. RIGA is essentially the repetition of this process until an admissible gate fidelity is obtained. The algorithm is shown to be monotonic in the sense that the infidelity that is measured by the Lyapunov function is nonincreasing along all the steps of the algorithm. Furthermore,  strong convergence results of RIGA are available  for controllable closed systems~\cite{PerSilRou19}. Namely, for $T_f$ big enough, RIGA converges to an exact solution of the quantum gate generation problem. Furthermore, numerical experimentations have  shown that RIGA generates small control inputs with a bandwidth that contains the natural frequencies of the system, at least if the control seed ${\overline u}^0(t)$ does not contain unnecessary high frequencies. However, it must obey some generic
conditions that are fulfilled for control profile including enough  harmonics of small amplitude\footnote{See section  \ref{sSeed} about the choice of the seed  ${\overline u}^0(t)$ for the proposed algorithm.}.

The algorithm that will be presented in this work was obtained directly from RIGA based on  a Fock-Liouville representation of  open quantum  systems.
 When re-transformed back into its original representation (of a Lindblad-Master equation) this algorithm have exhibited nice
 physical interpretations, including the presence of the adjoint Lindblad-Master equation (in the Heisenberg picture).
 We have chosen to present the results directly in its final form,  we shall not present here how it can be obtained from RIGA \footnote{The reader may refer to \cite{CBA2024} for these aspects of RIGA as well as a comparison of RIGA and GRAPE.}.

 So the algorithm  which is the main contribution of this work can be applied for open control systems described by Lindblad Master equations  on a Hilbert space of arbitrary dimension $n$ and  for a quantum gate of  arbitrary dimension $\nbar \leq n$. The first part of each step $\ell$ of this new algorithm consists in integrating backwards $\nbar^2$ copies of the adjoint Lindblad equation  from the final conditions $J_\sigma(T_f), \sigma \in \Lambda$ (observables in the Heisenberg  picture).  The second part of the algorithm consists in integrating $\nbar^2$ copies of the Lindblad Master equations with initial conditions $\rho_\sigma(0), \sigma \in \Lambda$ in closed loop with a Lyapunov based tracking control law. The $\nbar^2$ final conditions  $J_\sigma(T_f)$ and the $\nbar^2$ initial conditions
$\rho_\sigma(0), \sigma \in \Lambda$
 are projectors onto adequate pure states such that  the quantum gate operations are ensured in an analogous way that is considered in quantum tomography context \cite{NieChu11}. The adequate Lyapunov function  for the tracking control is:
 \[
    {\mathcal V}(t)  =   \nbar^2  - \sum_{\sigma \in \Lambda} \trace \left( J_\sigma (t) \rho_\sigma(t)  \right)
 \]
 It is easy to show that $\mathcal V(T_f)$ corresponds to the sum of  final individual gate infidelities of each member of the collection of $\nbar^2$ systems.
 Inside each step $\ell$, $\mathcal V(t)$ is nonnegative, nonincreasing  and $\mathcal V(T_f)$ it is equal to zero if and only if $\rho_\sigma(T_f) = J_\sigma(T_f), \sigma \in \Lambda$, that is, the quantum gate was exactly generated. Furthermore, this algorithm is monotonic in the sense that  the sequence  defined by the $\mathcal V(T_f)$ that are obtained along the successive steps is nonincreasing.

 In this paper, we also show how to include a clock control that allows to incorporate an extra (virtual) control which may be useful for finding an ``optimal'' final time $T_f$ of the gate.
 The clock control is in fact a virtual input $v_0$ that controls the running of a virtual time $\tau$ according to the differential law $\frac{d t}{d \tau} = {(1+v_0(t))}$. Since the algorithm presented here is an adaptation of RIGA for open systems described by a Lindblad-Master equations, and since
 RIGA admits strong convergence results for the  controllable unitary case, one expects that, when there exists a  control $u^*$ achieving exactly the gate and when the first variation around  the trajectory associated to $u^*$ is somehow controllable,  the algorithm  will converge to this set $u^*$, at last locally. This last conjecture will be the subject of a future research.

 We also show in this paper that this algorithm may be also regarded as an iterative method that converges to  a control input  that satisfies  the first order stationary conditions of an optimal control problem associated to the cost function ${\mathcal V}(T_f)  $.
 It also important to stress that the algorithm structure  is naturally adapted for array processors that could deal with each of the $\nbar^2$ copies of the system. Furthermore, the use of GPUs is strongly indicated since all the operations (including the 4-th order Runge-Kuta integration scheme) relies
 on the multiplication and the sum of $n\times n$ matrices. There is no gradient computation in the process, which seems to be useful in the application on the control design of quantum systems. The numerical stability is ensured by the stability of both the adjoint system and the one of the original Lindblad equation.

 Two examples of confined Cat-Qubit gates taken from \cite{Maz14,GauSarMir22} are presented.
 A first example of a $Z$-gate, recovering the existence of an optimal final time $T_f^*$ that was obtained in  \cite{Maz14} with  an adiabatic constant control. In this first example the fidelity of the adiabatic control is not far from the one that was produced by our algorithm, at least when both gate-times coincides to $T_f^*$. A second example of a CNOT-gate of much greater dimension is also studied, showing also an optimal $T_f^*$. However, for the second example, the shape of the control pulses are much more important. The problem of optimising the gate-time of constant adiabatic control produces an inferior fidelity (and different gate-time) than the problem of optimising both the shape and gate-time that is considered by our algorithm. The infidelity of the results of the constant adiabatic control with optimal $T_f$ is $29.2\%$ higher than the one of our algorithm.


 The paper is organised as follows. In section \ref{sLindblad} we state some notations and also some known results about the Lindblad master equation and its adjoint formulation. In section \ref{sQGGP} we state the quantum gate generation problem in the context of the density operators that appear in Lindblad equations. In section \ref{sAlgorithm} we present the algorithm when the gate time is prescribed and  show that the  algorithm is  an iterative method converging to the first order stationary conditions of an optimal control problem. In section \ref{sClock} we show how one can include the clock-control in this algorithm, which is useful for optimizing the final time $T_f$ of the gate.  Finally, in section \ref{sExample} we shall present two worked examples for  cat-qubit gates along with the results of numerical experiments with the proposed algorithm.

 Throughout  this work, we assume that the underlying Hilbert-spaces are   of finite dimension. However as the chosen formulation uses the language of operators, the various formulas and algorithms must certainly admit a meaning in infinite dimension with suitable choices of functional spaces.

\section{The Lindblad-Master equation}
\label{sLindblad}

 We will consider an open quantum system that is described by a Lindblad Master equation \cite{Bre02}. We recall that the state of a Lindblad Master equation is a $n \times n$ density matrix $\rho(t)$ which is a positive definite hermitian matrix of unitary trace. Let $u =(u_1, \ldots, u_m) \in \RR^m$, and let $\rho$ be a density matrix in $\CC^{n \times n}$. We may define the super-operator\footnote{To avoid notation confusion in the sequel with indice $i$  we use the bold symbol $\iii$ for  $\sqrt{-1}$.}
\begin{equation*}
 \mathcal{L}_{u} (\rho)= -{\iii} \left[H_0 + \sum_{k=1}^{m}  u_{k} H_k, \rho \right]+
  \sum_{q=1}^{p} L_q \rho L_q^\dag - \frac{1}{2} \left\{ L^\dag_q L_q, \rho \right\}
\end{equation*}
The Lindblad master equation will be denoted by:
 \begin{equation*}
 \frac{d\rho(t)}{dt} = \mathcal{L}_{u(t)} (\rho(t))
\end{equation*}
We shall also consider the adjoint Lindblad-Master equation. The state $J(t)$ of the adjoint Lindblad Master equation, called \emph{Observable},  is a $n \times n$ hermitian  matrix $J(t)$. In our algorithm, such observable $J(t)$,  will be such that its spectrum  is always  contained in $[0, 1]$.
The adjoint Lindblad equation given below considers in fact the Heisenberg view-point of quantum mechanics with the super-operator
\begin{equation}
\label{eKu}
 {\mathcal L}^*_u ( J ) =-{\iii} \left[-H_0- \sum_{k=1}^{m}  u_{k} H_k, J \right] +\\
 \sum_{q=1}^{p} L_q^\dag J L_q - \frac{1}{2} \left\{ L^\dag_q L_q, J \right\}
\end{equation}
defining the adjoint Lindblad equation \cite{Bre02}:
 \begin{equation*}
 \frac{ dJ(t)}{dt} = {\mathcal{L}}^*_{u(t)} (J(t))
\end{equation*}
It is well known that, if one takes an observable $J(0)$ and compute the solution of
the Lindblad-Master equation $\rho(t)$, and after that one computes the expectation value of the observable, this entire process is equivalent
to take the initial condition of the state and compute the expectation value of the observable $J(t)$. In other words:
\begin{equation}
\label{eStar}
\trace( J(0) \rho(t)) = \trace(J(t)  \rho(0))
\end{equation}
The condition  \eqref{eStar}  will be important in order to ensure that our algorithm is monotonic.

\section{Quantum gate generation problem }
\label{sQGGP}

As said in the introduction, for closed quantum system and considering the unitary evolution of the Schrodinger equation \eqref{eSh},
the quantum gate generation problem can be stated as being the steering problem of Def. \ref{Def1}. For open quantum systems, in the case
where the state is a density operator, the quantum gate generation problem may be defined in a way that is similar to the quantum Tomography context
\cite{NieChu11}. This consists in constructing a set of pure states assuring the complete definition of the gate.

\begin{definition} \label{dD1}
 (Quantum Gate Generation Problem)
 Consider that
  $\{|e_i\rangle, i=1, \ldots, \nbar \}$ and $\{|f_i\rangle, i=1, \ldots, \nbar \}$ are two orthonormal subsets of $\CC^{n}$
 with $\nbar \leq n$. Let
  \begin{itemize}
     \item
        $|e_{ijR}\rangle =  \frac{1}{\sqrt{ 2}} ( |e_i\rangle + |e_j\rangle)$, $i > j$, and
        $|e_{ijI}\rangle = \frac{1}{\sqrt{ 2}} (|e_i\rangle + \iii |e_j\rangle)$, $i > j$
         \item
        $|f_{ijR}\rangle =  \frac{1}{\sqrt{ 2}} ( |f_i\rangle + |f_j\rangle)$, $i > j$, and
        $|f_{ijI}\rangle = \frac{1}{\sqrt{ 2}} (|f_i\rangle + \iii |f_j\rangle)$, $i > j$.
  \end{itemize}
 The quantum gate generation problem consists in
 finding  a set of $m$ control pulses $u:[0, T_f] \rightarrow \RR^m$ such that:
 \begin{itemize}
   \item [(i)] The state $\rho(t)$ is steered from the  state
     $|e_{i}\rangle \langle e_{i}|$ at $t=0$ to   the state $|f_{i}\rangle \langle f_{i}|$, at $t=T_f$ for $i =1, \ldots, \nbar$

  \item [(ii)] One must also steer all the $|e_{ijR}\rangle \langle e_{ijR}|$ at $t=0$ to  $|f_{ijR}\rangle \langle f_{ijR}|$
        at $t=T_f$.

 \item [(iii)] One must also steer all the $|e_{ijI}\rangle \langle e_{ijI}|$ at $t=0$ to  $|f_{ijI}\rangle \langle f_{ijI}|$
        at $t=T_f$.

  \end{itemize}

 \end{definition}
  We stress that all the final and initial conditions that defines the gate are pure states.
  We shall consider a set of $\nbar^2$ multi-indices $\sigma$ elements of
  \[
  \Lambda = \left\{i, ijR, ijI :  i,j \in \{1, \ldots , \nbar\}, i > j \right\}.
  \]
  for indexing the above  family of initial and final conditions.

\section{Monotonic  algorithm with a prescribed gate-time }
\label{sAlgorithm}
This section is devoted to the description of the proposed algorithm.

\subsection{Main  definitions}

Decompose the control input $u$ as $ u(t)= {\overline u}(t) +  {\widetilde u}(t)$ where ${\overline u}(t)$ appears in $\nbar^2$ copies of the (minus) adjoint Lindblad Master equation with different final conditions indexed by
$\sigma \in \Lambda$:
\begin{subequations}
\label{eReferenceLindblad}
\begin{eqnarray}
  \frac { d J_\sigma}{dt} (t) & = & -{\mathcal L}^*_{{\overline u}(t)} \left({J}_\sigma (t)\right)\\
  J_\sigma(T_f) & =  & \Pi_{|\phi_\sigma\rangle} = |\phi_\sigma\rangle \langle \phi_\sigma|,~ \sigma \in \Lambda
\end{eqnarray}
\end{subequations}
where
\[
  |\phi_\sigma\rangle = \left\{
 \begin{array}{l}
|f_i\rangle \mbox{, if $\sigma = i \in \{1, \ldots, \bar n\}$}\\
\frac{|f_i\rangle + |f_j\rangle}{\sqrt{2}}, \mbox{ if $\sigma = ijR, i,j \in \{1, \ldots, \bar n\}$}, i >j \\
\frac{|f_i\rangle + \iii |f_j\rangle}{\sqrt{2}}, \mbox{ if $\sigma = ijI, i,j \in \{1, \ldots, \bar n\}$}, i >j
 \end{array}
 \right.
\]
Then $u={\overline u}+  {\widetilde u}$ appears in   $\nbar^2$ copies of the system with different initial conditions:
\begin{subequations}
\label{eClosedLindblad}
\begin{eqnarray}
  \frac { d{\rho}_\sigma(t)}{dt} & = & \mathcal{L}_{{\overline u}(t) +  {\widetilde u}(t)} \left({ \rho}_\sigma(t)\right)\\
   { \rho}_\sigma (0)& = & \Pi_{|\epsilon_\sigma\rangle}= |\epsilon_\sigma\rangle \langle \epsilon_\sigma|,~ \sigma \in \Lambda
  \end{eqnarray}
\end{subequations}
where
\[
 |\epsilon_\sigma\rangle = \left\{
 \begin{array}{l}
|e_i\rangle \mbox{, if $\sigma = i \in \{1, \ldots, \bar n\}$}\\
\frac{|e_i\rangle + |e_j\rangle}{\sqrt{2}}, \mbox{ if $\sigma = ijR, i,j \in \{1, \ldots, \bar n\}$}, i >j \\
\frac{|e_i\rangle + \iii |e_j\rangle}{\sqrt{2}}, \mbox{ if $\sigma = ijI, i,j \in \{1, \ldots, \bar n\}$}, i >j
 \end{array}
 \right.
\]

\subsection{The Lyapunov Function}

We shall apply a Lyapunov based feedback law in the input $\widetilde u$ to be described in the sequel.
For this, consider the Lyapunov function:
 \begin{eqnarray*}
 \mathcal V(t) & =  & \nbar^2 - \sum_{\sigma \in \Lambda} \trace( {J}_\sigma(t) {\rho}_\sigma(t))
\end{eqnarray*}
where $J_{\sigma}(t)$ and $\rho_\sigma(t)$ are solution of~\eqref{eReferenceLindblad} and~\eqref{eClosedLindblad}, respectively.
 It is clear $\mathcal V(T_f)$ is the sum of the gate-infidelities of the members of the system
for all $\sigma \in \Sigma$ since
\[
 \trace ( \Pi_{|\phi_\sigma\rangle} \rho ) = \trace (|\phi_\sigma\rangle \langle \phi_\sigma| \rho) =  \langle \phi_\sigma| \rho |\phi_\sigma\rangle
 \]
 which is  the fidelity of $\rho$ with respect to the pure state $|\phi_\sigma\rangle$.

The following proposition explains why $\mathcal V(t)$ is a convenient Lyapunov function for the quantum gate generation problem.
\begin{proposition} Take $\sigma\in\Lambda$ and  $(J_\sigma,\rho_\sigma)$ solutions of~(\ref{eReferenceLindblad},\ref{eClosedLindblad}). Then necessarily  for any $t\in[0,T_f]$,  $0\leq \trace( {J}_\sigma(t) {\rho}_\sigma(t))\leq 1$. Thus $\mathcal{V}(t) \geq 0$.
Moreover if $\mathcal{V}(T_f)=0$ then for all $\sigma\in\Lambda$,  $\rho_\sigma(T_f)= \ket{\phi_\sigma}$.
\end{proposition}
The proof of this proposition just relies on the fact that for all $t\in[0,T_f]$, the spectrum of $J_\sigma(t)$ belongs to~$[0,1]$ according to~\cite{SepulSR2010} and that $\rho_\sigma(t)$ is a density operator.

Computing the time-derivative of the Lyapunov function, one gets:
\begin{eqnarray*}
\frac{d \mathcal V(t)}{dt}  & = & - \sum_{\sigma \in \Lambda} \trace\left(\frac{d J_\sigma}{dt}(t) {\rho}_\sigma(t) +  {J}_\sigma(t)  \frac{d{ \rho}_\sigma(t}{dt})\right) \\
           & = & - \sum_{\sigma \in \Lambda} \trace( -{\mathcal L}^*_{\overline u} ({J}_\sigma(t)) {\rho}_\sigma(t)) + \trace( {J}_\sigma(t) {\mathcal L}_{(\overline u + \widetilde u)} ( {\rho}_\sigma(t)))\\
           & = & -  \sum_{k=1}^m {\widetilde u}_k  \left[\sum_{\sigma \in \Lambda} {\trace (J_\sigma(t)  [-\iii H_k, \rho_\sigma(t)])} \right]
\end{eqnarray*}
In the computations above we have used the fact that ${\mathcal L}_u (\rho)$ is affine in $u$ and
 \[
 \trace (-{\mathcal L}^*_{\overline u} ({J}_\sigma) {\rho}_\sigma) + \trace( {J}_\sigma {\mathcal L}_{\overline u} ( {\rho}_\sigma)) = 0.
 \]
Assume that $[0,T_f]\ni t \mapsto \overline{u}(t)$ is given, then  the Lyapunov-based control with strictly positive gain  $g_k>0$ provides $\widetilde{u}$ via  the  forward integration of~\eqref{eClosedLindblad}:
\begin{eqnarray*}
u_k(t)  =  {\overline u}_k(t) +\underbrace{  g_k \left[\sum_{\sigma \in \Lambda} \trace (J_\sigma(t)  [-\iii H_k, \rho_\sigma(t)]) \right]
}_{\text{\normalsize ${\widetilde u}_k = g_k { F}_k(t)$} }
\end{eqnarray*}
By construction :
\[
\frac {d{\mathcal V}(t)}{dt} = - \sum_{k=1}^m g_k {{ F}_k}^2(t) \leq 0
\]
and so the Lyapunov function is nonincreasing inside a specific step $\ell$ of the algorithm.\\[0.3cm]

\subsection{The algorithm }
We are ready to state the main contribution of this paper which is the following algorithm.\\[0.3cm]
\textbf{BEGIN ALGORITHM}
\begin{itemize}
\item [$\sharp 1.$] Choose the seed input  ${\overline u}^0:[0, T_f] \rightarrow \RR^m$.\\
\textbf{BEGIN STEP $\ell$}.\\
  \begin{itemize}
 \item[$\sharp 2.$] Set ${\overline u}(\cdot) = {\overline u}^{\ell-1}(\cdot)$. \\
             Integrate (backwards) in $[0, T_f]$ the $\nbar^2$ copies of the adjoint
             system:\\
                 \begin{eqnarray*}
                  \frac { d J_\sigma}{dt} (t) & = & -{\mathcal L}^*_{{\overline u}(t)} \left({J}_\sigma (t)\right), \;\;\;
                  J_\sigma(T_f)  =   \Pi_{|\phi_\sigma\rangle} = |\phi_\sigma\rangle \langle \phi_\sigma|, \sigma \in \Lambda
                \end{eqnarray*}

\item [$\sharp 3.$]  Integrate (forward) in $[0, T_f]$ the $\nbar^2$ copies of the system in closed loop:

                    \begin{eqnarray*}
                  \frac {d{ \rho}_\sigma(t)}{dt} & = & \mathcal{L}_{u(t)} \left({ \rho}_\sigma(t)\right), \;\;\;
                   { \rho}_\sigma (0) =  \Pi_{|\epsilon_\sigma\rangle}, \sigma \in \Lambda \\
                  u_k(t) & = & {\overline u}_k(t) + \underbrace{\sum_{\sigma \in \Lambda} \trace (J_\sigma(t)  [-\iii H_k, \rho_\sigma(t)])}_{{\widetilde u}_k(t) = g_k F_k(t), ~ g_k >0}, k=1, \ldots, m
                  \end{eqnarray*}

                  Set ${\overline u}^{\ell}(\cdot) = u(\cdot) $ (closed loop input). \\
                   Notice that control  constraints can be included here  just by imposing   that   each  $u_k(t)$  remains between $u_k^{\min}$ and $u_k^{\max}$. \\
\item [$\sharp 4.$] If the final fidelity  is acceptable,
             then terminate the algorithm. Otherwise, execute step $\ell+1$.
 \end{itemize}
 \item [\  ] \textbf{END STEP $\ell$}
 \end{itemize}
 \textbf{END ALGORITHM}

\begin{theorem} \label{t1} The value of the Lyapunov function  $\mathcal V(T_f)$ obtained in the end of the step $\ell-1$ of the previous algorithm is equal to the initial value $\mathcal V(0)$ for step $\ell+1$.
In particular the Lyapunov function is non-increasing along all the steps of the algorithm.
\end{theorem}

\begin{proof}
 In $\sharp 3.$ of step $\ell$ we integrate (forward):
\begin{eqnarray*}
 \frac{d {  \rho}_\sigma(t)}{dt} & = & {\mathcal{L}}_{u(t)} \left({ \rho}_\sigma(t)\right), \;\;\;
   { \rho}_\sigma (0) =  \Pi_{|\epsilon_\sigma\rangle}, \sigma \in \Lambda \\\
  u(t) & = & {\overline u}^\ell(t)
 \end{eqnarray*}
We stress that all the initial condition are the same $\rho_\sigma(0), \sigma \in \Lambda$ for all $\ell =1,2, \ldots$.
Now note that, in $\sharp 2.$ of step $\ell+1$, we integrate backwards
 \begin{eqnarray*}
  \frac { d J_\sigma}{dt} (t) & = & -{\mathcal L}^*_{{\overline u}(t)} \left({J}_\sigma (t)\right), \;\;\;
  J_\sigma(T_f)  =   \Pi_{|\phi_\sigma\rangle} = |\phi_\sigma\rangle \langle \phi_\sigma|, \sigma \in \Lambda\\
    \overline u(t) & = & {\overline u}^\ell(t)
\end{eqnarray*}
Note that, as we integrate backwards  $-{\mathcal L}^*_{{\overline u}(t)}$, that is, the time reversing of the adjoint system, then ${ J}_\sigma(0)$
plays the role of the final observable, whereas ${ J}_\sigma(T_f)$ plays the role of the initial observable.
Then the Schrodinger/Heisenberg duality gives:
\[
\underbrace{\trace({ J}_\sigma(0) \rho_\sigma(0))}_{\mbox{defines $\mathcal V(0)$ for step $\ell+1$}} = \underbrace{\trace(J(T_f) { \rho}_\sigma (T_f))}_{\mbox{defines $\mathcal V(T_f)$ for step $\ell$}}
\]
\end{proof}

\subsection{Choice of the  seed}
\label{sSeed}

We state now some remarks about the choice of the seed ${\overline u}^0(t)$ of the proposed algorithm.
We may choose an integer $M > 0$,   a  period $T > 0$ and  small amplitude $\mathbf{A} > 0$ and
two vectors $\mathbf a$, $\mathbf b$ in $\RR^{m~M}$ with  \textbf{random}, independent, and uniformly distributed
entries in $[-1,1]$:
\begin{eqnarray*}
\mathbf a & = & \{a_{k,\ell}, k=1,\ldots,m, \ell=1, \ldots, M\}\\
\mathbf b & = & \{b_{k,\ell}, k=1,\ldots,m, \ell=1, \ldots, M\}
\end{eqnarray*}
So define the seed $\overline{u}^0$    from a given  reasonable  initial  control   $[0,T_f]\ni t \mapsto u^{init}(t)$ perturbed as follows:
\begin{equation}
  \label{seed}
  \overline{u}^0_k (t) = u^{init}_k (t)+\mathbf{A} \left\{ \sum_{\ell = 1}^{M} \left[  a_{k \ell} \sin(2\ell \pi t / T) + b_{k \ell} \cos(2\ell \pi t / T) \right]\right\}
 \end{equation}
Such  choices are inspired from    a  mathematical  result given in~\cite{PerSilRou19} for purely controllable Hamiltonian dynamics and  requiring that the seed must contain the presence of sufficient harmonics  in order to guarantee the convergence to an exact gate generation. For $M$ big enough, this convergence is ensured with probability one with respect to the random variables $\mathbf a$, $\mathbf b$. This is due to the fact that the proof is based on the main ideas of the Coron's return method~\cite{Cor07}.

\subsection{Optimal control interpretation}
\label{sOptimal}

In this section we shall show that our algorithm converges to a control law that obeys the stationary conditions of first order of an optimal control problem.  We consider the   notations used in~(\ref{eReferenceLindblad},\ref{eClosedLindblad}) for the quantum  states $\Pi_{|\epsilon_\sigma\rangle}$ and  observables
 $\Pi_{|\phi_\sigma\rangle}$, $\sigma \in \Lambda$ defining  a quantum gate.
 We will consider the same set of $\nbar^2$ copies of the system:
\begin{equation}
\label{eq:rhosig}
    \frac{d { \rho}_\sigma(t)}{dt} = {\mathcal L}_u (\rho_\sigma(t)) = {\mathcal L}_0 (\rho_\sigma(t)) + \sum_{k=1}^m u_k [-\iii H_k, \rho_{\sigma}(t)]
,~
\rho_\sigma(0) = \Pi_{|\epsilon_\sigma\rangle}, \sigma \in \Lambda
\end{equation}
 Consider the following  optimal control problem:
 \begin{center}
 \em
     Find $u : [0, T_f]\rightarrow \RR^m$ in order to minimise:
$
 \nbar^2 - \sum_{\sigma \in \Lambda} \trace ( \Pi_{|\phi_\sigma\rangle} \rho_\sigma(T_f) )
$
subject to
$
 \frac{d{\rho}_\sigma(t)}{dt} = {\mathcal L}_u (\rho_\sigma(t)),
~
\rho_\sigma(0) = \Pi_{|\epsilon_\sigma\rangle}, \sigma \in \Lambda
.
$
 \end{center}
 Consider the Lagrangian
\[
   \nbar^2 - \sum_{\sigma \in \Lambda} \trace ( \Pi_{|\phi_\sigma\rangle} \rho_\sigma(T_f)) + \sum_{\sigma \in \Lambda} \int_0^{T_f} \trace\left(J_\sigma(t)\big({\mathcal L}_u (t) - \frac{d{\rho}_\sigma}{dt}(t)\big)\right) dt
\]
with the adjoint operators $J_\sigma(t)$.
The stationary conditions of this Lagrangian versus any variation  $\delta \rho_\sigma(t)$ such that $\delta \rho_\sigma(0)=0$ yield to the adjoint
system for $J_\sigma$ with its  final conditions:
\begin{equation}
\label{eq:Jsig}
  \forall \sigma \in \Lambda,~\forall t \in [0, T_f],~ \frac { d J_\sigma}{dt} (t) =-{\mathcal L}^*_{{ u}(t)} \left({J}_\sigma (t)\right),~
  J_\sigma(T_f) = \Pi_{|\phi_\sigma\rangle} = |\phi_\sigma\rangle \langle \phi_\sigma|
\end{equation}
The stationary conditions of this Lagrangian versus any variation  $\delta u(t)$ yield to
\begin{equation}
\label{eq:uvar}
\forall t \in [0, T_f],~\forall k \in \{1, \ldots, m\},~F_k(t) = \sum_{\sigma \in \Lambda} \trace (J_\sigma(t) [-\iii H_k, \rho_\sigma(t)]) =0,
.
\end{equation}
Equations~\eqref{eq:rhosig}, \eqref{eq:Jsig} and~\eqref{eq:uvar} are thus the first order stationary conditions of the above optimal control problem.

Returning  to our algorithm, recall that:
\begin{itemize}
\item The sequence ${\mathcal V}_\ell = \mathcal V(T_f)$ is nonnegative and nondecreasing along the steps $\ell = 1, 2, \ldots$ of the algorithm.
Hence, this sequence must converge to some $\mathcal V^* \geq 0$.

\item Recall that $\frac{d{\mathcal V}}{dt} = - \sum_{k=1}^{m} g_k F_k^2(t) $ with $\widetilde{u}_k(t)= g_k F_k(t)$.

\item Then $\underbrace{{\mathcal V}_{\ell-1}}_{\mathcal V(0)} - \underbrace{{\mathcal V}_\ell}_{\mathcal V(T_f)} = \sum_{k=1}^{m} \int_{0}^{T_f} g_k  F_k^2(t) dt \rightarrow 0$ in the compact interval $[0, T_f]$.

\item It is easy to show that $\frac{d{F_k}}{dt}$ is bounded (because $\frac {d\rho_\sigma}{dt}$ and $\frac {d {J}_\sigma}{dt}$ are bounded)  and so the $F_k, F_k^2$ are bounded and Lipchitz continuous.

\item In particular, $F_k(t) \rightarrow 0$ in the sup norm for all $k =1, \ldots, m$ when the iteration step $\ell$ tends to infinity. This results from the monotonicity of the algorithm.

\end{itemize}
Thus the monotonic algorithm converges to some $u$ satisfying the   first order stationary conditions \eqref{eq:rhosig}, \eqref{eq:Jsig} and~\eqref{eq:uvar}, which is equivalent to say that the feedback $\widetilde u$  converges to zero as the iteration number $\ell$ goes to the infinity.

\section{Monotonic algorithm with   gate-time optimization}
\label{sClock}

The idea of ``controlling the clock'' appears for instance in \cite{FLMOR97}
in the context of ``orbital flatness'' (see the references therein for a control historical perspective).
 By controlling the clock we mean that we can introduce a virtual time $\tau$ such that
$\frac{dt}{d\tau} = (1 + v_0(\tau))$ where $v_0$ is the control of the clock.
Basically the virtual time $\tau$ can run  faster ($v_0 < 0$) or
 slower ($v_0 > 0$)
 than the real time $t$. This procedure ensures the existence of a new (virtual) control $v_0$ for the system.
 By controlling the clock one also changes the final time $T_f$ that is associated to the control problem.
 Denote:
\[
\begin{array}{rcl}
 \mathcal{L}_{0} (\rho) &  =  & -{\iii} \left[H_0, \rho \right] +
  \sum_{q=1}^{p} L_q \rho L_q^\dag - \frac{1}{2} \left\{ L^\dag_q L_q, \rho \right\}
\end{array}
\]
\[
\begin{array}{rcl}
 {\mathcal{L}}_{k} (\rho) &  =  & -{\iii} \left[H_k, \rho \right], k=1, \ldots, m
\end{array}
\]
Then the Lindblad master equation reads:\\
\[
 \frac{d \rho}{dt} = {\mathcal{L}}_{0} (\rho) + \sum_{k=1}^{m} u_k(t)  {\mathcal{L}}_{k} (\rho)
\]
The  virtual time $\tau$ produced  by a clock-control $v_0(\tau)$  via
\[
\frac{dt}{d\tau} = (1 + v_0(\tau)) = \alpha(\tau).
\]
cannot  reverse the time direction. Thus  we include the restriction $| v_0(\tau)| < 1$.
Note that the real time $t(\tau)$ is given by
$t(\overline\tau) = \int_0^{\overline \tau} \alpha(\tau) d\tau$.
Then, defining $v_k(\tau)  = (1 + v_0(\tau))u_k[t(\tau)]$ for $k\in\{1,\ldots,m\}$, one obtains:
\begin{eqnarray*}
\frac{{ d\rho}}{d\tau} & = & \frac{{ d \rho}}{dt} \frac{dt}{d\tau}  = (1 + v_0(\tau))({\mathcal{L}}_{0} (\rho) + \sum_{k=1}^{m} u_k {\mathcal{L}}_{k} (\rho)) \\
& = & \left( {\mathcal{L}}_{0} (\rho) +  \sum_{k=0}^{m} { v}_k(\tau) {\mathcal{L}}_{k} (\rho) \right)
\end{eqnarray*}
with $m+1$ virtual scalar control inputs
\[
v(\tau) = (v_0(\tau), v_1(\tau), \ldots, v_m(\tau))\in \RR^{m+1}.
\]
Denote:
\[
\begin{array}{rcl}
 {\mathcal L}^*_{0} (\rho) &  =  & {\iii} \left[H_0, \rho \right] +
  \sum_{q=1}^{p} L_q^\dag \rho L_q - \frac{1}{2} \left\{ L^\dag_q L_q, \rho \right\}
\end{array}
\]
And define the adjoint superoperator
\[
\begin{array}{rcl}
 {\mathcal L}^*_{k} (\rho) &  =  & {\iii} \left[H_k, \rho \right], k=1, \ldots, m
\end{array}
\]
 Then, given the system with clock control:
\[
 \frac{d \rho}{d\tau} = {\mathcal{L}}_{0} (\rho) + \sum_{k=0}^{m} v_k(\tau)  {\mathcal{L}}_{k} (\rho)
\]
Its corresponding adjoint system is given by:
\[
 \frac{d J}{d \tau} = {\mathcal L}^*_{0} (J) + \sum_{k=0}^{m} v_k(\tau)  {\mathcal L}^*_{k} (J)
\]
Then we may state the algorithm equipped with clock control\\[0.3cm]
\textbf{BEGIN ALGORITHM} \\ \
[$\sharp 0.$] Choose $T_{f}^{(0)}$ and the \textbf{seed input} $\overline u^0: [0, T_{f}^{(0)}] \rightarrow \RR^{m}$.  \\
\begin{itemize}
 \item[\  ]\textbf{BEGIN STEP $\ell$}.
    \begin{itemize}

      \item[$\sharp 1.$]  Let $T_f = T_{f}^{(\ell-1)}$. Define ${\overline v} : [0, T_f] \rightarrow \RR^{m+1}$, with ${\overline v}_0(\tau) = 0$ and ${\overline v}_k(\tau) = {\overline u}_k^{\ell-1}(\tau),  \forall \tau \in [0, T_f]$, $k=1, \ldots, m$.

     \item[$\sharp 2.$] Set ${\overline v}(\tau) = {\overline v}^{\ell-1}(\tau)$.
     Integrate (backwards) the  $\nbar^2$ copies of the adjoint system in $[0, T_f]$.
     Obtain the trajectories $J_\sigma(\tau)$,  for $\tau \in [0, T_f]$ with final condition $J_\sigma(T_f) = \Pi_{|\phi_\sigma\rangle}$, $\sigma \in \Lambda$.

   \item [$\sharp 3.$]   Integrate  the  $\nbar^2$ copies of the  closed loop system in $[0, T_f]$.\\
     Obtain the trajectories $\rho_\sigma(\tau)$,  for $\tau \in [0, T_f]$ with initial condition $\rho_\sigma(0) = \Pi_{|\epsilon_\sigma\rangle}$.
    Set ${\overline v}_k^{\ell}(\tau) =   {\overline v}_k(\tau) + g_k\sum_{\sigma \in \Lambda} \trace (J_\sigma(\tau) \mathcal{L}_k (\rho_\sigma(\tau)))$, $k=0, 1, \ldots, m$ and gains $g_k >0$.

   \item [$\sharp 4.$]     Compute the new final time $T_{f}^{(\ell)} = \int_{0}^{T_f} (1 + {\overline v}^\ell_0(\tau)) d\tau$.\\
   Compute ${\overline u}_k^\ell(t ({ \tau})) = \frac{ {\overline v}^\ell_k( \tau)}{ 1 + {\overline v}^\ell_0(\tau)}$, for $\tau \in [0, T_f]$, k=1, \ldots , m \\
   where $t ({ \tau})=\int_{0}^{ \tau} (1 + v_0(\tau')) d\tau'$.

    \item [$\sharp 5.$] If the final fidelity  is acceptable,
                 then terminate. Otherwise, execute step $\ell+1$.

    \end{itemize}
 \item[\   ] \textbf{END STEP $\ell$}.
 \end{itemize}
 \textbf{END ALGORITHM}.

The corresponding optimal control problem is then
 \begin{center}
 \em
     Find $T_f >0$ and $u : [0, T_f]\rightarrow \RR^m$ in order to minimise:
$
 \nbar^2 - \sum_{\sigma \in \Lambda} \trace ( \Pi_{|\phi_\sigma\rangle} \rho_\sigma(T_f) )
$
subject to
$
 \frac{d{\rho}_\sigma(t)}{dt} = {\mathcal L}_u (\rho_\sigma(t)),
~
\rho_\sigma(0) = \Pi_{|\epsilon_\sigma\rangle}, \sigma \in \Lambda
.
$
 \end{center}
 The first order stationary conditions~\eqref{eq:rhosig}, \eqref{eq:Jsig} and~\eqref{eq:uvar} have to be completed by the  following  condition relative to the variation of  $T_f$:
\begin{equation}
\label{eq:Tfvar}
     \sum_{\sigma \in \Lambda} \trace\big( \Pi_{|\phi_\sigma\rangle} \mathcal{L}_u(\rho_\sigma(T_f)) \big)=0
     .
\end{equation}
It is then clear that the above monotonic algorithm including clock-control always converges to some  $u$ and $T_f>0$ satisfying these  first order stationary conditions.

\section{Numerical simulations for confined cat-qubit gates}
\label{sExample}

{ The next two gate generations  are  taken from \cite{Maz14,GauSarMir22} combining  dissipative dynamics towards the code space   with  adiabatic Hamiltonian dynamics  for  cat-qubit gates.

\subsection{ $Z$-gate}}
Define the dissipation super operator:
\[
\mathcal D [L] \rho = L \rho L^\dag - \frac{1}{2} \{L^\dag L, \rho\}.
\]
Following~\cite[equation (5)]{Maz14} consider the following Lindblad  equation
\[
\frac{d \rho(t)}{dt} =  \kappa_2 \mathcal D [a^2 - \alpha^2] \rho + \kappa_1 \mathcal D [a] \rho -\iii u [ H_c , \rho]  \]
where  $a$ is the annihilation operator of an harmonic quantum oscillator with infinite dimensional Hilbert space with Fock Hilbert basis $\Big(| n\rangle\Big)_{n\geq 0}$: $a | n\rangle = \sqrt{n} | n-1\rangle$;
$\alpha$ is  any real number ($\alpha^2\equiv \alpha^2 I$ with $I$ identity operator), and the
 control Hamiltonian associated to the scalar control $u$  is given by $H_c= (a + a^\dag)$.

We have chosen  $\kappa_2=1$ and  $\kappa_1=1/100$.
We have taken $\alpha =2$ and we have truncated the Hilbert basis  up to $n^{\max}=20$ Fock states (including $|0\rangle$).
The definition of  coherent state  $|\alpha\rangle$ reads:
\[
|+\alpha \rangle  = \exp(\frac{-|\alpha|^2}{2}) \sum_{n=0}^\infty \frac{\alpha^n}{\sqrt{n!}} \; \; | n\rangle
\]
We define respectively the even and odd parity cats of norm one:  \\
\[
|C^{\pm}_{\alpha}\rangle \propto {|+\alpha \rangle  \pm |-\alpha \rangle }
\]
From a unitary evolution point of view, our quantum gate may be defined as the following steering problem:
steer $|C^{+}_{\alpha}\rangle$ to $|C^{-}_{\alpha}\rangle$, and
steer  $|C^{-}_{\alpha}\rangle$ to $|C^{+}_{\alpha}\rangle$.\footnote{This unitary operation is equivalent to the Z-gate considering $|0 \rangle \approx | \alpha \rangle$ and $|1 \rangle \approx |-\alpha \rangle$. In other words, in terms of the unitary operation we want to
steer $|e_1\rangle=|0 \rangle$ to $|f_1\rangle=|0 \rangle $  and  to
steer $|e_2\rangle=|1 \rangle$ to $|f_2\rangle=-|1 \rangle $.}

Define:
$|e_{12R} \rangle = \frac{|e_1 \rangle + |e_2 \rangle}{\sqrt 2}$,
$|e_{12I} \rangle = \frac{|e_1 \rangle + \iii |e_2 \rangle}{\sqrt 2}$,
$|f_{12R} \rangle = \frac{|f_1 \rangle + |f_2 \rangle}{\sqrt 2}$,
$|f_{12I} \rangle = \frac{|f_1 \rangle + \iii |f_2 \rangle}{\sqrt 2}$.
{As in Definition \ref{dD1}, the operations defining the Z-gate in the context of density matrices are:}
steer $|e_1 \rangle \langle e_1|$ to $|f_1 \rangle \langle f_1|$ at $t=T_f$,
steer $|e_2 \rangle \langle e_2|$ to $|f_2 \rangle \langle f_2|$ at $t=T_f$,
steer $|e_{12R} \rangle \langle e_{12R}|$ to $|f_{12R} \rangle \langle f_{12R}|$ at $t=T_f$ and steer $|e_{12I} \rangle \langle e_{12I}|$ to $|f_{12I} \rangle \langle f_{12I}|$ at $t=T_f$.

 For realising this gate at $t=T_f$, we have a well known nice constant (turning) adiabatic control:
\begin{equation}\label{eq:adiabcontrol}
  [0,T_f] \ni t \mapsto u(t)=u_{ad} \equiv \frac{\pi}{4 T_f \alpha}
\end{equation}
ensuring, when $\kappa_1=0$, an almost perfect gate for $T_f$ large enough. When $\kappa_1>0$, $T_f$ cannot be chosen too large (typically $T_f \kappa_1\ll 1$) in order to avoid the decoherence due to photon losses at rate $\kappa_1$.
In all simulations, the seed input ${\overline u}^0(t)$ is considered to be the adiabatic control slightly perturbed
by the sum of harmonics
\begin{equation}
\label{eSeed}
{\overline u}^0(t) = u_{ad}(t) + \mathbf{A} \sum_{\ell = 1}^{M} \left[ a_{k \ell} \sin(2\ell \pi t / T) + b_{k \ell} \cos(2\ell \pi t / T)\right],
\end{equation}
with $T=T_{f}^{(0)}$, and $ \mathbf{A} = \frac{|u_{ad}|}{100}$, $M=3$.
The control gain is $g_1=1$, $k=1, \ldots, m$. The gain of the clock-control is $g_0=0.1$.

Given a pure state $\Pi_{|\xi \rangle} = |\xi \rangle \langle \xi |$ with $|\xi \rangle \in \CC^n$ and a density matrix  $\rho$,  the  fidelity function is:
\[
\mbox{Fidelity}(\rho, \Pi_{|\xi \rangle}  ) = \langle \xi| \; \rho \; |\xi  \rangle
\]
Recall that the final conditions defining the quantum gate are given by pure states:
 \[
 J_\sigma(T_f) = \Pi_{ | \phi_\sigma \rangle} = | \phi_\sigma \rangle \langle \phi_\sigma |, \sigma \in \Lambda.
 \]
 The  gate infidelity $\mathcal I$ in the end of each step of the algorithm will be defined by
\begin{equation}\label{eq:infidelity}
\mathcal I = \max_{\sigma \in \Lambda} \left\{1 -  \mbox {Fidelity} \left( \rho_\sigma(T_f),  | \phi_\sigma \rangle \langle \phi_\sigma | \right) \right\}
=  \max_{\sigma \in \Lambda} \left\{1 - \langle \phi_\sigma| \; \rho_\sigma(T_f) \; |\phi_\sigma  \rangle \right\}
\end{equation}
which is the worst case infidelity for each  trajectory $\rho_\sigma(t), \sigma \in \Lambda$.

The simulation results with the clock-control monotonic algorithm  are summarised in figures \ref{fA} and \ref{fB}. Figure  \ref{fA} (left side) shows that, for $T_{f}^{(0)}= 5$ the final time converges to
some value that is close to $0.85$. For $T_{f}^{(0)}= 0.5$, the final time seems to converge to the same value, close to $0.85$. For $T_{f}^{(0)}=0.85$, the
values of $T_f$ along the steps of the algorithm remains always close to the initial value. Figure \ref{fA} (right side) shows that
the final infidelity is improved a lot by the algorithm for $T_{f}^{(0)}=5$, a little bit for $T_{f}^{(0)}=0.5$ and almost nothing for $T_{f}^{(0)}=0.85$ which
seems to be very close to the ``optimal'' final time. It must be stressed that the infidelity for the (unperturbed) adiabatic control for $T_f=0.85$ (without running the algorithm) is  $0.0696$. After running 80 iterations of the algorithm with $T_{f}^{(0)}=0.85$, the gate infidelity is $0.0669$, a really
small improvement with respect to the (unperturbed) adiabatic control conceived with the optimal final time. Thus the main interest in this example is to find the best value of $T_f$.

Figure \ref{fB} shows that, in all cases, the control pulses generated by the algorithm seems to converge to the same control pulses (the difference is almost indistinguishable in the (top) figure \ref{fB}. Note that we have introduced a saturation of $u$ between $-0.8$ and $+0.8$, which is compatible with the algorithm. The algorithm produces a small improvement of the adiabatic control conceived with the optimal final time by augmenting the control effort close to the beginning and to the end of the interval $[0, T_f]$. From the findings in last figure, one may say that the more important here is to optimise the time of the gate. In this case the (unperturbed) adiabatic control is almost so efficient than the ``optimal' control that is generated by our algorithm. The bottom of Figure \ref{fB} depicts the seeds in the three cases $T_{f}^{(0)} = 5$, $T_{f}^{(0)} = 0.5$ and $T_{f}^{(0)} = 0.85$. They are in fact the slightly perturbed adiabatic controls for each correspondent case.

One can run this algorithm for more complete models that includes for instance the ``buffer''  cavity, and then
one may tune the final time of adiabatic control. This could be done even in the case that the quantum gate is generated by a
series of different adiabatic control pulses

\begin{figure}
    \centering{\includegraphics[scale=0.7]{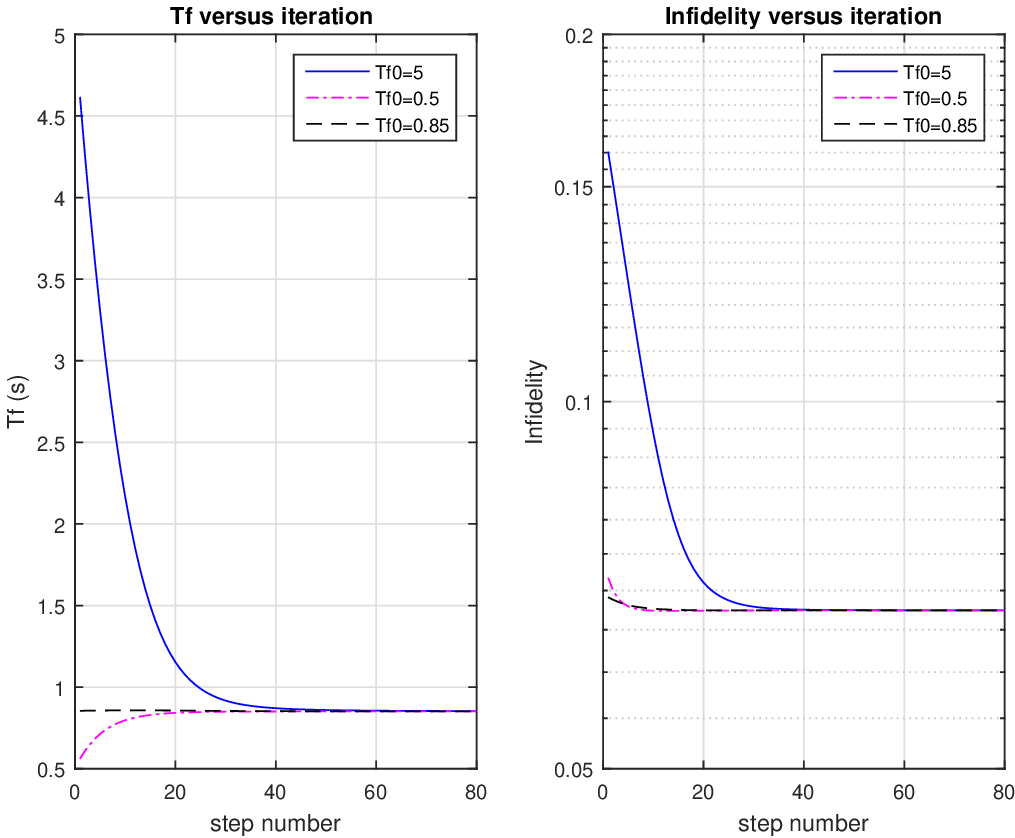}}
    \caption{Z-gate generation. Left side: evolution of $T_f$  and of infidelity $\mathcal I$ (see~\eqref{eq:infidelity}). In these simulations $T_{f}^{(0)} = 5$,
    $T_{f}^{(0)} = 0.85$, and $T_{f}^{(0)}=0.5$ ($\kappa_1=1/100$). Right side: infidelity that is obtained for each iteration of the algorithm for
    $T_{f}^{(0)} = 0.85$, and $T_{f}^{(0)}=0.5$.
    }
    \label{fA}
\end{figure}

\begin{figure}
    \centering{\includegraphics[scale=0.7]{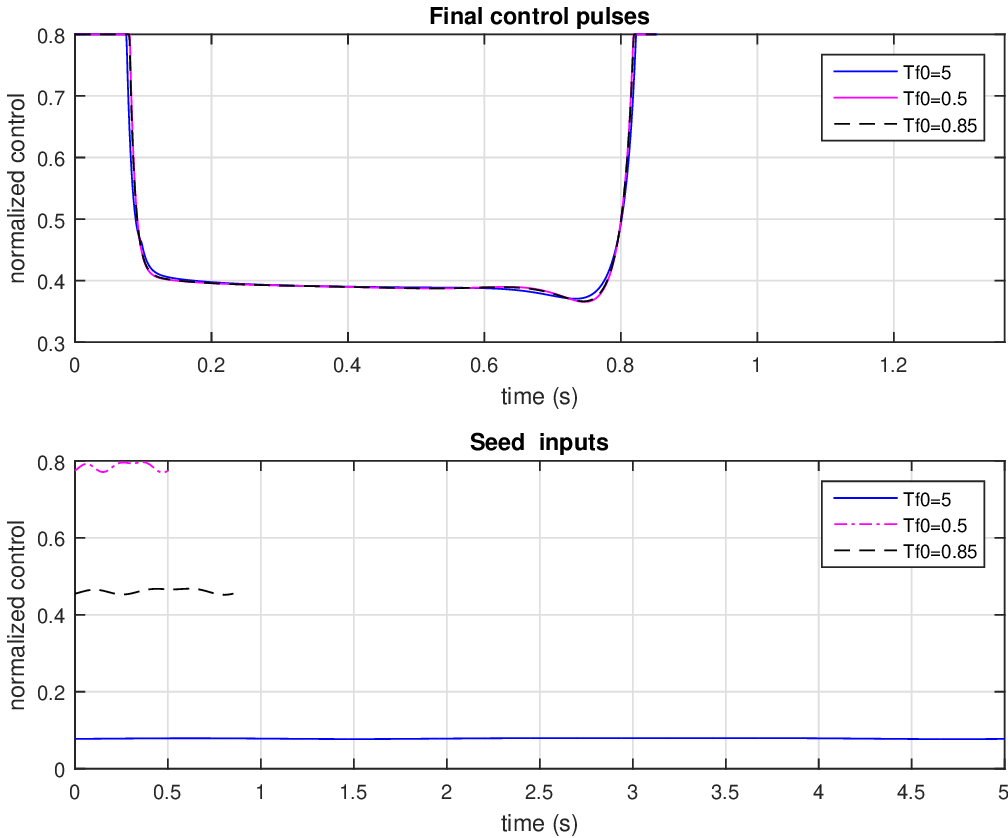}}
    \caption{Z-gate generation. Top: Final control pulses  for $T_{f}^{(0)} = 5$, for  $T_{f}^{(0)} = 0.85$, and for $T_{f}^{(0)} =0.5$ ($\kappa_1=1/100$). Botton: The seed inputs for each case, $T_{f}^{(0)} = 5$, for  $T_{f}^{(0)} = 0.85$, and for $T_{f}^{(0)} =0.5$ (seed inputs are slightly perturbed  constant adiabatic control~\eqref{eq:adiabcontrol}).}
    \label{fB}
\end{figure}

{

\subsection{ CNOT-gate}

The    notations similar to those of  previous  sub-section  are   used  here. Following~\cite[equation (16)]{GauSarMir22} where we have added single photon losses of rate $\kappa_1$, the master equation describing the evolution  with scalar control input $u(t)$  is given by
\begin{eqnarray*}
\frac{d \rho}{dt} & = & -\iii u \big[ (a_{co}+a^{\dag}_{co}-2 \alpha I_{co}) \otimes (a^\dag_{ta}a_{ta} - \alpha^2I_{ta})\otimes  I_{qu}~, ~ \rho\big]
\\& & - \iii g_2 \big[ (a^2_{co}  - \alpha^2)\otimes I_{ta}\otimes \ket{e}\!\bra{g} + ((a^\dag)^2_{co}  - \alpha^2)\otimes I_{ta}\otimes  \ket{g}\!\bra{e}~,~ \rho\big] \\
  & & +  k_2 {\mathcal D}[(a^2_{co} - \alpha^2 I_{co})\otimes I_{ta}\otimes I_{qu} ] \rho
  \\
  && + k_1 {\mathcal D}[a_{co}\otimes I_{ta}\otimes I_{qu}] \rho+  k_1 {\mathcal D}[I_{co}\otimes a_{ta}\otimes I_{qu}] \rho
\end{eqnarray*}
where $\alpha^2=4$, $k_2=1$, $k_1=\frac{1}{1000}$, $g_2=10$ and $k_2 T_f \approx 1$.  The underlying Hilbert-space is the  tensor of three Hilbert-space $\mathcal{H}_{co}\otimes \mathcal{H}_{ta}\otimes \mathcal{H}_{qu}$: Hilbert-space of the control cat-qubit $\mathcal{H}_{co}$, Hilbert-space of the target cat-qubit $\mathcal{H}_{ta}$, Hilbert-space of an ancillary  qubit $\mathcal{H}_{qu}\equiv \mathbb{C}^2$. The operators $a_{co}$ and $a_{ca}$ are the annihilation operators of the control
and target cat-qubits. Denote by $|\alpha/-\alpha\rangle_{co/ta} \approx |0_L/1_L\rangle_{co/ta}$ the coherent states for the control and target cat-qubits, and by $\ket{g}$ and $\ket{e}$  the ground and excited states of the  qubit. In the context of unitary transformations, the CNOT-gate
is  the operator that maps $e_1= |0_L\rangle_{co} \otimes |0_L\rangle_{ta}  \otimes \ket{g}$ to
$f_1 = |0_L\rangle_{co} \otimes |0_L \rangle_{ta}  \otimes\ket{g}$,  $e_2= |0_L\rangle_{co} \otimes |1_L\rangle_{ta}  \otimes\ket{g}$ to
$f_2 = |0_L\rangle_{co} \otimes |1_L \rangle_{ta}  \otimes\ket{g}$, $e_3= |1_L\rangle_{co} \otimes |0_L\rangle_{ta}  \otimes\ket{g}$ to
$f_3 = |1_L\rangle_{co} \otimes |1_L \rangle_{ta}  \otimes\ket{g}$ and $e_4= |1_L\rangle_{co} \otimes |1_L\rangle_{ta}  \otimes\ket{g}$ to
$f_4 = |1_L\rangle_{co} \otimes |0_L \rangle_{ta}  \otimes\ket{g}$.
As in Definition \ref{dD1}, the operations that define  the CNOT-gate in the context of density matrices are:\\
(i) steer $|e_i \rangle \langle e_i|$ to $|f_i \rangle \langle f_i|$ at $t=T_f$, for $i=1, \ldots, 4$;\\
(ii)  steer $|e_{ijR} \rangle \langle e_{ijR}|$ to $|f_{ijR} \rangle \langle f_{ijR}|$ at $t=T_f$ and\\
(iii) steer $|e_{ijI} \rangle \langle e_{ijI}|$ to $|f_{ijI} \rangle \langle f_{ijI}|$ at $t=T_f$
for all $i, j \in \{1, \ldots, 4\}$ with $i > j$.

Simulations shows that\footnote{This property that ensures that conditions (ii) and (iii) are not important is due to particular symmetries if the system. See also~\cite{Goerz2014} for interesting and  connected results.} only the conditions (i) are sufficient to generate the gate up to a very good precision. The improvement of the infidelity of considering both conditions (i) and (ii) is less that $1\%$ but the computation effort  is four times greater. All the presented simulations considers only the conditions (i) for the construction of the Lyapunov function. The (in)fidelity that is presented  considers only the set of conditions (i), but the complete set of conditions (i), (ii) and (iii) are considered for computing the final (in)fidelity, which we call by\footnote{The infidelity correction when one computes the worst case of all the conditions (i), (ii), (ii) with respect to the infidelity computed only with the set of conditions (i) is  approximately given by $0,8 \%$ in all cases.} ``corrected-infidelity''.
It well  known that a nice (constant)  adiabatic control for generating a CNOT gate at $t=T_f$ is given by $u=u_{ad}=\frac{\pi}{4 \alpha T_f}$.

The proof of Theorem \ref{t1} shows that the Lyapunov function of the end of Step $\ell$ is equal to the Lyapunov function of the begining of step $\ell +1$. However the numerical integration  induces and error that is reflected by a difference of these values. This numerical difference can be used to estimate the numerical error of the Runge-Kutta integration. This can be used to find a convenient time-step of the integration. Also, the number of levels of both resonators for the truncated models were estimated to be at least $n=17$ for $\alpha =2$. A greater value of $\alpha$ will certainly need a greater $n$ to in order to ensure the same precision. This means that the underlying Hilbert-space has  dimension $n \times n \times 2 = 578$. So the $578\times 578$-density matrix of the system represents a state of dimension $334 \, 084$ for the Runge-Kutta integration of the ODE system
 For the next results we have considered a a slightly perturbed adiabatic control defined for $T_f^{(0)}=1.5$ as the seed of the algorithm with the same form \eqref{eSeed}  with $T=T_{f}^{(0)}$,  $\mathbf{A} = \frac{2|u_{ad}|}{1000}$, $M=3$.

At the top of Figure~\ref{FC}
the evolution of the infidelity $\mathcal I = 1 - \mathcal F$ s presented along the steps of the algorithm.
The final (uncorrected) infidelity  $0.0009784$ is attained in step 1249 of the algorithm.
The final corrected\footnote{Corrected in the sense that it considers the worst case infidelity of all the 16 conditions (i) and (ii) of definition \ref{dD1}.}  infidelity at the step 1249 is $0.0009865$. The variation between the corrected and the uncorrected value is $ \approx 0.84\%$.
At the bottom of same figure, the evolution of the gate time $T_f$ is presented. A final $T_f = 1.259 s$ is attained in step 1259 of the algorithm. At the top of Figure \ref{FC} it is presented the final control pulse that is constructed in step 1259. It can be compared with the constant adiabatic control for the same gate time $T_f=1.259 s$ and the seed input of the algorithm. It is interesting to say that the definite integral of the final control pulse and the constant adiabatic control on $[0, T_f]$ is respectively given and $3.92721$ and $3.92699$, corresponding to a variation of only $0.6\%$. However, the infidelity of the constant adiabatic pulse defined with the gate time  $T_f=1.259 s$ is $0.00143$, which is more that $40\%$ bigger that the one of our final control pulse that is produced by our algorithm. In some sense, our final control pulse is an optimised adiabatic control, since the integral is almost conserved. The bottom of figure \ref{FD}
shows the final fidelity of the constant adiabatic control as a function of the gate time $T_f$.  The optimal value of the infidelity of the constant adiabatic control shown in Figure \ref{FD} is close to $0.00128$ (corresponding to $T_f \approx 1.8 s$), which is  $29.2\%$ greater than (corrected) infidelity of the final ``optimal'' control pulse. Anyway, it is clear that the pulse shape is rather important in this second example. For the second example, the problem of optimising the fidelity with respect to the gate-time  of a constant adiabatic control (which concerns to plot of the botton of Figure \ref{FD}) gives a rather different final infidelity when one optimises both the gate-time and the shape of the control input, which concerns to the result of our algorithm depicted in Figure \ref{FC}.

\begin{figure}
    \centering{\includegraphics[scale=0.7]{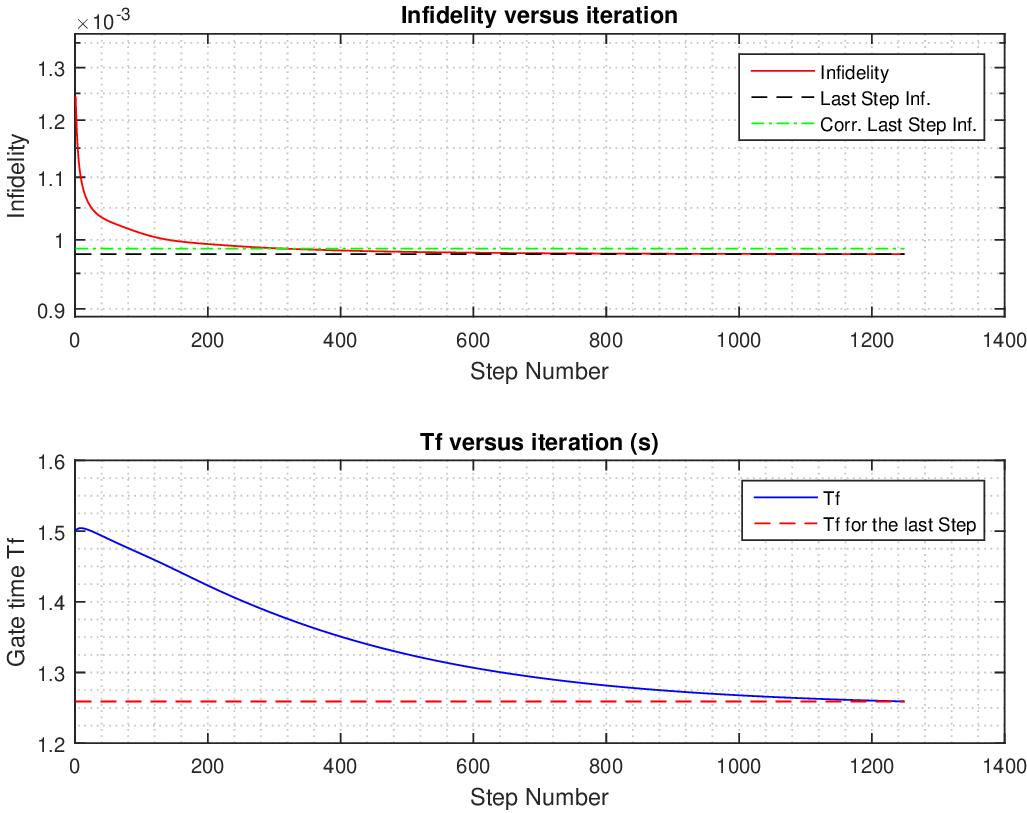}}
    \caption{CNOT-gate. Top: Evolution of gate-infidelity along the steps of the algorithm. The (uncorrected) infidelity (see~\eqref{eq:infidelity}) of the last step is given by $0.9784\times 10^{-3}$. The corrected infidelity of the last step is $0.98655\times 10^{-3}$ (the variation between the corrected and the uncorrected value is $ \approx 0.84\%$). Botton: The evolution of the gate-time $T_f$ along the steps of the algorithm. The gate-time for the last step is $T_f = 1.259 $.}
    \label{FC}
\end{figure}

\begin{figure}
    \centering{\includegraphics[scale=0.7]{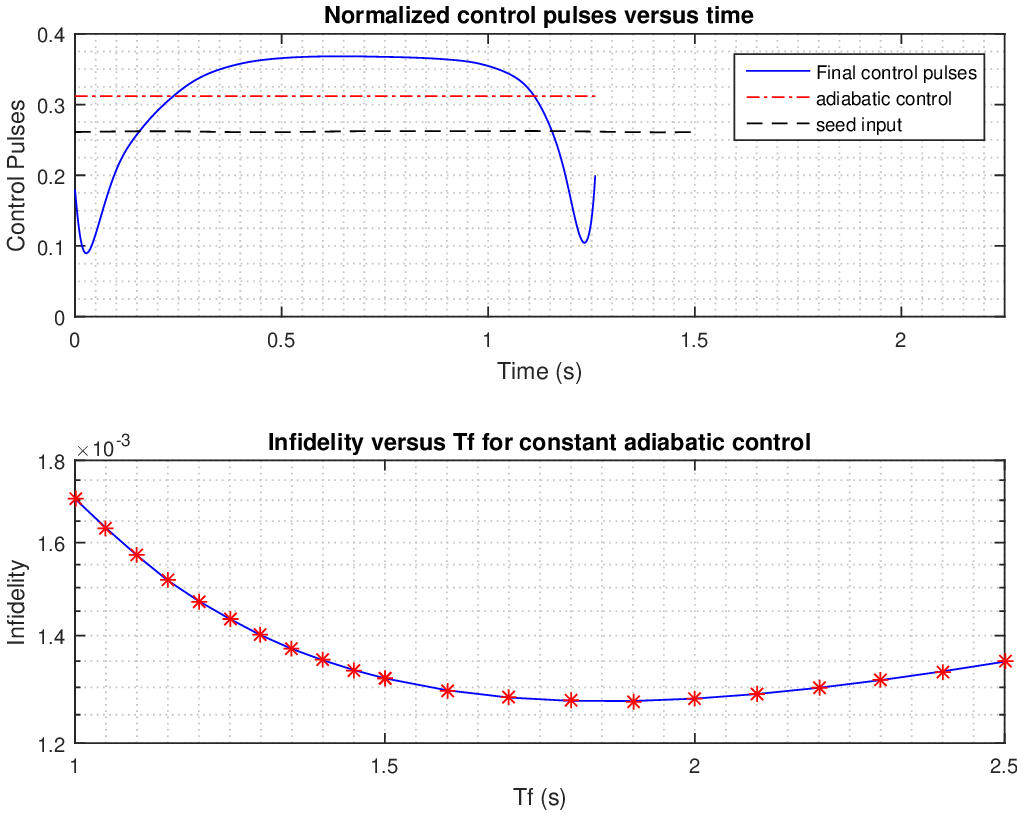}}
    \caption{CNOT-gate. Top: Final control pulses generated by the algorithm of section~ref{sClock}. The constant adiabatic control for the same $T_f$ of the final control pulses is also presented along with the seed of the algorithm, which is a slightly  perturbed adiabatic control defined  for a initial gate-time $T_{f}^{(0)} = 1.5 s$. Botton: Infidelity that was obtained after simulating a constant adiabatic control presented as a function of gate-time $T_f$. The minimum infidelity is close to $0.00128$ for $T_f =1.8 $, whereas the infidelity obtained with algorithm is $0.00098655$ with a shorteroptimized time $T_f = 1.259 $ .}
    \label{FD}
\end{figure}

}

{
\section{Conclusions}
A monotonic numerical method for generating quantum gates for open systems is described in this work. This method is strongly related to a \emph{Smooth-Case} version of the Krotov method that considers \emph{sequential} input update, that is, the input is updated ``on the fly''  of the simulation of the system dynamics  \cite{Koch2016,Goerz2014}. Our algorithm is generalised in this work by including the so-called \emph{clock-control}, which was shown to be equivalent to an algorithm that converges to the stationary conditions of an optimal control problem that regards not only the shape of the control pulses, but also seeks an optimal gate-time $T_f$. The effectiveness  of this generalised algorithm was tested in two case-studies of physical interest (confined cat-qubit gates), showing promising  results and the ability of obtaining the optimal gate-time  and optimal control pulses for these two examples. The standard form of such algorithm called RIGA in~\cite{PerSilRou19} is different from the algorithm that is presented in this work. In   \cite{CBA2024}, the authors shows how one obtains, for Hamiltonian dynamics  the presented algorithm from RIGA  . It is also shown in \cite{CBA2024} that GRAPE is a kind
of discrete version of RIGA when the objective function is the Lyapunov function, and when there is no sequential update, that is, the  input is updated only in the end of each step of the algorithm.}

\def\cprime{$'$}

\appendix

\section{Numerical implementation }

The numerical integration of the systems and their adjoints considers
the equations of Section \ref{sAlgorithm}. All this is implemented as  matrix version of 4th-order Runge-Kutta (fixed step) method, which regards sums and multiplications of $n\times n$ matrices. We have implemented this in a MATLAB program, but it is clear that a GPU implementation could improve a lot the run-time of the method. As described in Section  \ref{sAlgorithm}, all the backwards integrated paths of the adjoint system  $J_\sigma(t_k), \sigma \in \Lambda$ for the discrete time $t_k = \frac{k T_f}{N_{sim}}, k=0, \ldots, N_{sim}$ must be saved in memory. For instance, in the simulation of the second example,  we have chosen $N_{sim}=1000$, where $N_{sim}$ is the number of Runge-Kutta steps, in order to ensure a good precision \footnote{As  said in Example 2, the numerical precision may be estimated by the numerical difference of the Lyapunov function of  the end of step $\ell$ and the one of the  of step $\ell+1$ which in theory must be exactly the same. The difference that is obtained is then caused by the error of the numerical integration.}.
For large dimensional systems, one may have a memory overflow. In order to avoid this, one could do the backwards integration of step $\sharp 2$ of the algorithm without saving all the data in memory, but only saving the ``final conditions'' $J_\sigma(0), \sigma \in \Lambda$, and then perform a forward integration using that data as initial conditions. The numerical stability is ensured for the adjoint system for the backwards integration because such system coincides with the Heisenberg point of view, and hence is stable (from a dynamic system perspective). However the forward integration of the adjoint system (which may be unstable from a dynamic system point of view) may present bad numerical properties in some examples. In order to overcome this, one may save for instance, the backwards integrated data $J_\sigma(t_k)$ for $k = 0, 10, 100, \ldots$ and  do the forward integration by reseting the initial conditions to these values for
$t_k, k=0, 10, 100, \ldots$. This will divide the memory space by a factor of $10$ and  will preserve the numerical precision as well.

As  a last remark, it is important to say that one may include input limitations in the algorithm. When there is no clock-control, one may do this by standard saturation of the control-law by its maximal  value.  It is clear that the seed input must also obey this restriction. When the clock-control is present, one may saturate the clock control ${\overline v}_0(t)$ in a standard way, but since the other controls ${\overline v}_k(t), k>0$ are virtual controls, the saturation of these controls must be corrected by  the factor ${1+ v_0(t)}$ that multiplies the real controls in order to compute the virtual controls. This is done in a way that the contribution of each virtual feedback ${\widetilde v}_k(t), k=0,1, \ldots, m$ to the derivative of the Lyapunov function is always non-positive.
Consider the saturation map:
\begin{equation}
\mbox{sat} (x, A, -B) = \left\{
\begin{array}{l}
\mbox{$x$, if $x \in [-B, A]$}\\
-B, \mbox{if $x < -B$}\\
A, \mbox{if $x  > A$}
\end{array}
\right.
\end{equation}
The saturated feedback that was used in the simulations of this work is of the form:
\begin{equation}
  \label{eSatControl}
{\widetilde v}_k(t) = \mbox{sat} (g_k F_k, A_k, -B_k), k=0,1, \ldots, m
\end{equation}
where ${\widetilde v}_k(t) = g_k F_k$ is the standard ``non-saturated'' feedback law, and
\[F_k =  \sum_{\sigma \in \Lambda} \trace (J_\sigma(\tau) \mathcal{L}_k (\rho_\sigma(\tau))),
k=0, 1, \ldots, m.
\]
Furthermore, $A_0=u^0_{max}$, $-B_0 = \max \{-u^0_{max}, -\frac{{\overline u}^{\ell-1} }{u_{k}^{max}}-1, \frac{{\overline u}^{\ell-1} }{u_{k}^{max}}-1, k=1, \ldots, m\}$ , and $A_k=B_k= (1+v_0) u_0^{max}, k=1, \ldots,m$.

This saturated control has a theoretical explanation given in the sequel. Recall
that the the control law is of the form (see Section \ref{sClock}) :\\
\[
 {\overline v}_k^\ell(\tau) = {\overline u}^{\ell-1}(\tau) + {\widetilde v}_k(\tau), k=1, \ldots, m
\]
It is easy to see that the derivative of the Lyapunov function is
\[
\frac{d}{d\tau} \mathcal V = -\sum_{k=0}^m {\widetilde v}_k(t) F_k(t)
\]
Then the following result holds:
\begin{proposition}
Assume that $|{\overline u}^{\ell-1} | \leq u_{k}^{max}$ for $u_{k}^{max}>0, k=0, \ldots m$. Consider a control law of the form \eqref{eSatControl}  such that ${\widetilde v}_k(t) F_k \geq 0$, in any circumstances\footnote{The signal of $F_k(\tau)$ is not known a priori.} and
  $|{\overline u}^{\ell-1} | \leq u_{k}^{max}, k=0, \ldots m$. Suppose that a given ${\widetilde v}_0$ maximises the product ${\widetilde v}_0(\tau) F_0(\tau)$. Assume that the other virtual feedback  ${\widetilde v}_k(\tau), k=1, \ldots, m$
  minimizes  $\frac{d\mathcal V}{d\tau}$ for this given $v_0$. Then $A_0=u^0_{max}$, $B_0 = - \max \{-u^0_{max}, -\frac{{\overline u}^{\ell-1} }{u_{k}^{max}}-1, \frac{{\overline u}^{\ell-1} }{u_{k}^{max}}-1 , k=1, \ldots , m\}$ ,  $A_k=(1+v_0) u_k^{max}-{\overline u}^{\ell-1}$,
  and $B_k = (1+v_0) u_k^{max}+{\overline u}^{\ell-1}$,  $k=1, \ldots,m$.
\end{proposition}

\begin{proof} Since ${\widetilde v}_k(\tau) F_k(\tau) \geq 0$ for all $k =0, \ldots, m$   in all circumstances without knowing the signal of $F_k (\tau)$ \emph{a priori}, this means that the intervals $[-B_k, A_k]$ must contain zero, otherwise we cannot choose freely the signal of ${\widetilde v}_k(\tau) F_k(\tau)$ when the algorithm is running. As
${\overline u}^\ell = \frac{{\overline u}^{\ell-1}(\tau) + {\widetilde v}_k(\tau)}{1+v_0}$,
and ${\overline u}^\ell \leq u_k^{max}$ this means that ${\widetilde v}_k \leq (1+v_0)  u_k^{max} - {\overline u}^{\ell-1} =A_k \geq 0$ and  ${\widetilde v}_k \geq -(1+v_0)  u_k^{max} - {\overline u}^{\ell-1} =-B_k \leq 0$. This implies that $v_0 \geq \frac{{\overline u}^{\ell-1} }{u_{k}^{max}}-1$ and  $v_0 \geq \frac{-{\overline u}^{\ell-1} }{u_{k}^{max}}-1$. The rest of the statement is straightforward and is left to the reader.
\end{proof}

\end{document}